\documentclass[twocolumn,superscriptaddress,showpacs,pra]{revtex4-1}
\usepackage{amssymb,bm, amsbsy, amsmath, amsthm, latexsym, dsfont, array, layout, graphicx, mathrsfs,color,soul}
\usepackage[colorlinks=true,citecolor=blue,linkcolor=blue,anchorcolor=blue,urlcolor=blue]{hyperref}
\hypersetup{linktoc=all}
\renewcommand{\qedsymbol}{$\blacksquare$}
\allowdisplaybreaks
\newcommand{\ket}[1]{\left|{#1}\right\rangle}
\newcommand{\bra}[1]{\left\langle{#1}\right|}
\newcommand{\braket}[2]{\langle{#1}|{#2}\rangle}
\newcommand{\ketbrad}[1]{\left|{#1}\rangle\!\langle{#1}\right|}
\newcommand{\ketbra}[2]{\left|{#1}\rangle\!\langle{#2}\right|}
\newcommand{\dslash}{\not{\text{d}}}
\newcommand{\red}{\color[rgb]{0.8,0.3,0.2}}
\newcommand{\blue}{\color[rgb]{0,0,0.6}}
\newcommand{\green}{\color[rgb]{0.0,0.7,0.0}}
\renewcommand{\l}{\left}
\renewcommand{\r}{\right}
\renewcommand{\a}{\alpha}
\renewcommand{\b}{\beta}
\renewcommand{\d}{\operatorname{d}}
\renewcommand{\t}{\perp}
\renewcommand{\S}{\rm S}
\renewcommand{\-}{\text{-}}
\renewcommand{\H}{\mathbb{H}}
\renewcommand{\L}{\mathbb{L}}
\newcommand{\D}{\mathcal{D}}
\newcommand{\E}{\mathcal{E}}
\renewcommand{\P}{\mathcal{P}}
\newcommand{\Q}{\mathcal{Q}}
\newcommand{\A}{\mathcal{A}}

\renewcommand{\url}[1]{}
\newcommand{\urlprefix}{}
\renewcommand{\href}[1]{}
\newcommand{\blk}{\color{black}}
\newcommand{\blu}{\color{blue}}
\newcommand{\kl}[1]{{\color{CHIGUSA} \small \bf [[{#1}]]}}

\usepackage{xcolor}
\definecolor{CHIGUSA}{RGB}{58,143,183}
\definecolor{AOTAKE}{RGB}{0,137,108}
\definecolor{HATOBANEZUMI}{RGB}{114,99,110}

\newcommand{\li}{\color{AOTAKE}}

\begin{document}
	
	\title{Divisibility hierarchy of open quantum systems}
	\author{Fei-Lei Xiong}
	\affiliation{Hefei National Laboratory for Physical Sciences at the Microscale, Department of Modern Physics, University of Science and Technology of China, Hefei, Anhui 230026, China}
	\affiliation{CAS Center for Excellence in QIQP and the Synergetic Innovation Center for QIQP, University of Science and Technology of China, Hefei, Anhui 230026, China}
	\author{Zeng-Bing Chen}
	\affiliation{Hefei National Laboratory for Physical Sciences at the Microscale, Department of Modern Physics, University of Science and Technology of China, Hefei, Anhui 230026, China}
	\affiliation{CAS Center for Excellence in QIQP and the Synergetic Innovation Center for QIQP, University of Science and Technology of China, Hefei, Anhui 230026, China}
	
	\begin{abstract}
In the theory of open quantum systems, divisibility of the system dynamical maps is related to memory effects in the dynamics. By decomposing the system Hilbert space as a direct sum of several Hilbert spaces, we study the relationship among the corresponding dynamical maps. It is shown that if the dynamical maps of the open system possess a chain of invariant subspaces, there exists a divisibility hierarchy for their corresponding dynamics. Two classes of examples are given for illustrating these hierarchical structures. One is the pure-dephasing dynamics, and the other is the decay dynamics. Our results offer a systematic approach to obtaining the divisibility conditions and non-Markovian witnesses for these dynamics. Moreover, as a new way of decomposing open quantum systems, it is worthy of further study.
	\end{abstract}
	
	
\maketitle
	
\emph{Introduction.}---In the physical world, a quantum system is usually inevitably coupled to other quantum systems called the environment~\cite{BP02}. As a result, the dynamics of the (open) system possesses some stochastic nature. In many cases, it can be dealt with Born-Markov approximation~\cite{WM07,VA17}, i.e., the dynamics can be approximated as memoryless (Markovian)~\cite{G76,K72,L76}. In the past decades, as the experimental conditions have undergone a great progress~\cite{VA17} and the study of open quantum systems has become deeper, the memory effects have been attracting more and more interest~\cite{VA17,BLPV16,RHP14,LHW17}. 

Unlike the clear classical definition of Markovianity---a Markovian process is the process whose future evolution only depends on the present state itself rather than the trajectory to it~\cite{BP02}, the quantum counterpart is still in debate~\cite{HCLA14,LHW17}, even though there have been various definitions, measures, and witnesses of quantum non-Markovianity~\cite{RHP14}. In this paper, we study the memory effects in open quantum systems based on \emph{divisibility}~\cite{RHP10,WC08}. The definition of quantum non-Markovianity based on this concept is motivated by finding the quantum analogy to the Chapman-Kolmogorov equation, which is often employed to describe classical Markov processes~\cite{GS01}. Moreover, it relates to many measures of quantum non-Markovianity~\cite{LHW17}, such as trace distance measure~\cite{BLP09,LPB10,B12}, RHP measure~\cite{RHP10}, 
decay rates measure~\cite{HCLA14}, etc.~\cite{CM14}. 
Assume that the dynamics of a system S is described by a family of dynamical maps $\{\E(t,0),t\geq0\}$~\cite{LHW17}, where $\rho_{\S}(t)=\E(t,0)\rho_{\S}(0)$, with $\rho_{\S}(t)$ and $\rho_{\S}(0)$ denoting the system state at $t$ and $0$, respectively. 
The definition of divisibility reads as follows: If for arbitrary $t_2\geq t_1\geq 0$, there exists a CPTP map $\Q(t_2,t_1)$ satisfying
\begin{align}\label{eq-def-div}
\E(t_2,0)=\Q(t_2,t_1) \E(t_1,0)\, ,
\end{align}
$\E(t,0)$ is divisible; otherwise, it is indivisible, i.e., the dynamics is non-Markovian~\cite{RHP14}.
 
	
By the above definition, determining whether the dynamics of a system S is non-Markovian requires full knowledge of the dynamical map $\E(t,0)$. However, in general, because of the complexity of the total system, one cannot obtain the exact form of $\E(t,0)$~\cite{BP02}. Approximations are usually required to analyze divisibility of the dynamics. One standard approach beyond Born-Markov approximation is the Nakajima-Zwanzig projection operator technique~\cite{N58,Z60,P17}, through which master equation of the form
\begin{align}\label{eq-nz}
\frac{d\rho_{\S}(t)}{d t}=\int_{0}^{t}d \tau \mathcal{K}(t-\tau)\rho_{\S}(\tau)\,
\end{align}
can be derived. The superoperator $\mathcal{K}(t-\tau)$ is usually called memory kernel. Even though Eq.~\eqref{eq-nz} possesses a memory kernel, 
it can always be transformed to an equation in the time-local form~\cite{CS79,STH77,ACH07}. 
The time-local equation is associated with a family of \emph{decoherence matrices} $\{\bm{d}(t), t\geq 0\}$
~\cite{HCLA14}. (In the following, we would abbreviate $\bm{d}(t)$ to $\bm{d}$.) Only if $\bm{d}$ is positive-semidefinite, i.e., $\bm{d}\geq0$, the dynamics is divisible. Otherwise, it is indivisible.

In this paper, we decompose the system Hilbert space as direct sum of several subspaces and investigate the properties of their corresponding dynamics. It is shown that when $\E(t,0)$ has a chain of invariant subspaces, the divisibility conditions form a hierarchy. We find two classes of dynamics satisfying the condition: one is the pure-dephasing dynamics; the other is decay dynamics. We give one explicit example for each situation and analyze their properties. Furthermore, we consider the more general case that the ``subdynamics'' is not trace-preserving. At last, we leave some open questions for the further study of open quantum systems in this way.
	
\emph{Theoretical structure.}---
The Hilbert space of a quantum system $\mathbb{H}_{\S}$ is usually divided as $\mathbb{H}_{\S}=\mathbb{H}_{ 1}\otimes \mathbb{H}_{ 2}\otimes\cdots\otimes \mathbb{H}_{ n}$, where $\mathbb{H}_i$ denotes the Hilbert space associated to the $i$th degree of freedom of the system. Here, we consider another possible decomposition: $\mathbb{H}_{\rm S}$ decomposed as the direct sum of several Hilbert spaces~\cite{A97}, i.e., $\mathbb{H}_{\S}=\mathbb{H}_{\a_1}\oplus\mathbb{H}_{\a_2}\oplus\cdots\oplus\mathbb{H}_{\a_n}$.


To simplify the discussion, we first consider $\mathbb{H}_{\rm S}=\mathbb{H}_{\a_1}\oplus\mathbb{H}_{\a_2}$ with $\operatorname{dim}\left( \mathbb{H}_{\a_1} \right) <\infty$. Then any operator $X_{\S}$ of the system can be expressed in the form $X_{\rm S}=\left(\begin{array}{c c}
	{X_{\a_1\a_1}}&{X_{\a_1 \a_2}} \\
	{X_{\a_2 \a_1}}&{X_{\a_2\a_2}} \\
	\end{array}\right)$, in which $X_{\a_i \a_j}$ denotes a matrix block.
All $X_{\S}$ and $X_{\a_1\a_1}$ constitute the linear spaces $\L(\H_{\S})$ and $\L(\H_{\a_1})$, respectively~\cite{A97}. Apparently, $\L_{\a_1}\subset \L_{\S}$ (Hereafter, we shall abbreviate $\L(\H_{\S})$ and $\L(\H_{\a_i})$ to $\L_{\S}$ and $\L_{\a_i}$, respectively.). 
There exists a linear space ${\L}_{\a_1^{\perp}}$ satisfying ${\L}_{\rm S}={\L}_{\a_1}\oplus {\L}_{\a_1^{\perp}}$. As a result, a superoperator of the system $\E:{\L}_{\rm S} \rightarrow {\L}_{\rm S}$ can be decomposed as $\E=\mathcal{E}_{\a_1\a_1}\oplus{\mathcal{E}_{\a_1 \a_1^{\perp}}}\oplus{\mathcal{E}_{\a_1^{\perp} \a_1}}\oplus{\mathcal{E}_{\a_1^{\perp}\a_1^{\perp}}}$, where
\begin{subequations}\label{eq-sub}
\begin{align}
\E_{\a_1\a_1}&=\mathcal{P}_{\a_1 \a_1} \E \mathcal{P}_{\a_1 \a_1}\, ,\label{eq-def-sub1}\\
\E_{\a_1\a_1^{\perp}}&=\mathcal{P}_{\a_1 \a_1} \E \mathcal{P}_{\a_1^{\perp} \a_1^{\perp}}\, ,\\
\E_{\a_1^{\perp}\a_1}&=\mathcal{P}_{\a_1^{\perp} \a_1^{\perp}} \E \mathcal{P}_{\a_1 \a_1}\, ,\\
\E_{\a_1^{\perp}\a_1^{\perp}}&=\mathcal{P}_{\a_1^{\perp} \a_1^{\perp}} \E \mathcal{P}_{\a_1^{\perp} \a_1^{\perp}}\, ,
\end{align}
\end{subequations}
in which $\mathcal{P}_{\a_1 \a_1}$ ($\mathcal{P}_{\a_1^{\perp} \a_1^{\perp}}$) denotes the projection operator of the linear space $\L_{\a_1}$ ($\L_{\a_1^{\perp}}$).  


		
Apply the above decomposition to the dynamical map of an open quantum system. Given $\E(t,0)$, one can construct $\mathcal{E}_{\a_1\a_1}(t,0)$ according to Eq.~\eqref{eq-sub}. 
Because $\E(t,0)$ is CPTP, $\E_{\a_1 \a_1}(t,0)$ is completely positive (CP)~\footnote{See the Appendix.}. Therefore, the family $\l\{\E_{\a_1\a_1} (t,0)  \r\}$ can be interpreted as the ``pseudodynamics'' of $\a_1$. By ``pseudodynamics'', we mean that $\l\{\E_{\a_1\a_1} (t,0)  \r\}$ generates a family of positive-semidefinite matrices $\{\rho_{\a_1 \a_1}(t)=\E_{\a_1\a_1}(t,0)\rho_{\a_1 \a_1}(0) \}$ for any density matrix $\rho_{\a_1 \a_1}(0)$.

For some systems, there exist decompositions ensuring that $\L_{\a_1}$ is a $\E(t,0)$-invariant subspace for $t\geq0$~\cite{RN12}. This ensures that $\mathcal{E}_{\a_1\a_1}(t,0)$ is trace-preserving for $t\geq0$~\footnotemark[\value{footnote}]. Thus, all $\mathcal{E}_{\a_1\a_1}(t,0)$ are CPTP, and the family $\l\{\E_{\a_1\a_1} (t,0)\r\}$ describes a conventional quantum dynamics. Assuming that $\E(t,0)$ is always invertible, and the inverse is denoted as $\E^{-1}(t,0)$, the divisibility of $\mathcal{E}_{\a_1\a_1}(t,0)$ and that of $\E(t,0)$ satisfy the following theorem.

\newtheorem{theorem}{Theorem}
\begin{theorem}\label{theorem1}
If $\E(t,0)$ is divisible and $\E_{\a_1 \a_1}(t,0)$ are all CPTP, then $\E_{\a_1 \a_1}(t,0)$ is also divisible.
\end{theorem}
\begin{proof}
According to the decomposition $\H_{\S}=\H_{\a_1}\oplus\H_{\a_2}$, the system density matrix can be written in the form $\rho_{\S}=\left(\begin{array}{c c}
{\rho_{\a_1\a_1}}&{\rho_{\a_1 \a_2}} \\
{\rho_{\a_2 \a_1}}&{\rho_{\a_2\a_2}} \\
\end{array}\right)$. 

When $\rho_{\S}=\rho_{\a_1\a_1}$, i.e., ${\rho_{\a_1 \a_2}}={\rho_{\a_2 \a_1}}={\rho_{\a_2\a_2}}=0$, following Eq.~\eqref{eq-sub},
\begin{align}\label{eq-the-1}
\rho_{\S}'=\E\rho_{\S}=(\E_{\a_1\a_1}+\E_{\a_1^{\t} \a_1}) \rho_{\a_1\a_1} =\left(\begin{array}{c c}
{\rho_{\a_1 \a_1}'}&{\rho_{\a_1 \a_2}'} \\
{\rho_{\a_2 \a_1}'}&{\rho_{\a_2\a_2}'} \\
\end{array}\right)\,,   
\end{align}
where $\rho_{\a_1 \a_1}'=\mathcal{E}_{\a_1\a_1} \rho_{\a_1\a_1}$.
Assume $\E$ and $\E_{\a_1\a_1}$ are CPTP maps. 
One can easily deduce that $\rho_{\a_2\a_2}'= 0$. Moreover, $\rho_{\a_1 \a_2}'=\rho_{\a_2 \a_1}'=0$ and $\rho_{\S}'=\rho_{\a_1 \a_1}'$~\cite{A97}. That is, $\mathcal{E} \rho_{\a_1\a_1}=\mathcal{E}_{\a_1\a_1} \rho_{\a_1\a_1}$. Following this and Eq.~\eqref{eq-the-1}, $\mathcal{E}_{\a_1^{\perp}\a_1}=0$. In this case, one can prove that $\mathcal{E}_{\a_1\a_1}$ and $\mathcal{E}_{\a_1^{\perp}\a_1^{\perp}}$ are invertible~(See the Appendix), and the inverse of $\E$ can be formulated as
\begin{align}\label{eq-inverse}
\mathcal{E}^{-1}=\mathcal{E}_{\a_1\a_1}^{-1}-\mathcal{E}_{\a_1\a_1}^{-1}\mathcal{E}_{\a_1 \a_1^{\perp}}\mathcal{E}_{\a_1^{\perp}\a_1^{\perp}}^{-1}+\mathcal{E}_{\a_1^{\perp}\a_1^{\perp}}^{-1}\, .
\end{align}
		
Because the dynamical map $\E(t,0)$ is divisible, 
\begin{align}
\mathcal{Q}(t_2,t_1)=\mathcal{E}(t_2,0)\mathcal{E}^{-1}(t_1,0)\
\end{align}
is a CPTP map. Thus, $\mathcal{Q}_{\a_1\a_1}(t_2,t_1)$ is CP~\footnotemark[\value{footnote}]. According to Eq.~\eqref{eq-inverse},
\begin{align}
\mathcal{Q}_{\a_1^{\perp}\a_1}(t_2,t_1)&=\mathcal{P}_{{\a_1}^{\t}} \E(t_2,0) \E^{-1}(t_1,0)  \mathcal{P}_{\a_1 \a_1} \nonumber\\
&=\mathcal{E}_{\a_1^{\perp} \a_1}(t_2,0) \mathcal{E}_{\a_1\a_1}^{-1}(t_1,0) \nonumber \\
&=0\, .\label{eqn:cona}
\end{align}
Thus, $\forall\, X_{\a_1 \a_1}\in \L_{\a_1}$,
\begin{align}
\operatorname{Tr}\l(X_{\a_1\a_1}\r)=\operatorname{Tr}(\mathcal{Q}(t_2,t_1)X_{\a_1 \a_1})= \operatorname{Tr}(\mathcal{Q}_{\a_1\a_1}(t_2,t_1)X_{\a_1 \a_1})\,,
\end{align}
i.e., $\mathcal{Q}_{\a_1\a_1}(t_2,t_1)$ is trace-preserving. Therefore, $\mathcal{Q}_{\a_1\a_1}(t_2,t_1)$ is CPTP. Explicitly,
\begin{align}
\mathcal{Q}_{\a_1\a_1}(t_2,t_1)&=\mathcal{P}_{\a_1 \a_1} \E(t_2,0) \E^{-1}(t_1,0)  \mathcal{P}_{\a_1 \a_1} \nonumber\\
&=\mathcal{E}_{\a_1\a_1}(t_2,0)\mathcal{E}_{\a_1\a_1}^{-1}(t_1,0)\, .\label{eq-div-alp}
\end{align}
To arrange Eq.~\eqref{eq-div-alp} in another way, 
\begin{align}\label{eq-comp-alpha}
\mathcal{E}_{\a_1\a_1}(t_2,0)=\mathcal{Q}_{\a_1\a_1}(t_2,t_1)\mathcal{E}_{\a_1\a_1}(t_1,0)\, .
\end{align}
That is, $\E_{\a_1 \a_1}(t,0)$ is divisible.
\end{proof}

Theorem~\ref{theorem1} offers us a method of finding witnesses of non-Markovianity. That is, when $\E_{\a_1 \a_1}(t,0)$ is indivisible, $\E(t,0)$ must be indivisible, i.e., the system dynamics is non-Markovian~\cite{RHP14}. 

Generally, there may be more than one non-trivial invariant subspaces of $\E(t,0)$. Correspondingly, we have the following corollary.	
\newtheorem{corollary}{Corollary}
\begin{corollary}\label{coro1}
Suppose there exists a chain ${\L}_{1}\subset {\L}_{2}\subset \cdots\subset {\L}_{n}\subset {\L}_{\rm S}$ of invariant subspaces of $\E(t,0)$ ($\forall t\geq 0$), where ${\L}_{k}=\H_{k}\times\H_{k}$, with $\H_k$ denoting some Hilbert space~\cite{BP02}. 
Then, their divisibility conditions form a hierarchy, i.e., the divisibility of the dynamics corresponding to a larger space implies the divisibility corresponding to a smaller space.
\end{corollary}

	
	
There are two classes of dynamics whose invariant subspaces can be found easily. One is the pure-dephasing dynamics and the other is the decay dynamics. In the following, we shall discuss these two classes of dynamics.

\emph{Example for the pure-dephasing case.}---Consider the boson-boson pure-dephasing model proposed in Ref.~\cite{XLC18}. The dynamical map $\E(t,0)$ in the interaction picture satisfies~\footnotemark[\value{footnote}]
\begin{align}
\E(t,0)\sum_{m,n=0}^{\infty}\rho_{mn} \ketbra{m}{n}=\sum_{m,n=0}^{\infty}\rho_{mn} \operatorname{e}^{\eta((m-n)t)} \ketbra{m}{n}\, ,
\end{align}	
where $\ket{i}$ and $\ket{j}$ stand for number states~\cite{WM07} and $\eta(t)=\eta^*(-t)$. 
Through Theorem~\ref{theorem1}, the divisibility condition of the dynamics can be derived and a family of non-Markovianity witnesses can be obtained.

Define 
\begin{align}\label{eq-def-Lk}
{\L}_k:=\operatorname{span}\l\{\ketbra{m}{n}:m,n\in \{0,1,\cdots,k-1\}\r\}\,,
\end{align}
where $k\geq 2$.
Then one can derive its corresponding pseudodynamical maps $\E_k(t,0)$ satisfying
\begin{align}
\E_k(t,0)\sum_{m,n=0}^{k-1} X_{mn} \ketbra{m}{n}=\sum_{m,n=0}^{k-1}X_{mn} {e}^{\eta((m-n)t)} \ketbra{m}{n}\, ,
\end{align}
where $X_{mn}\in \mathbb{C}$ and $\eta(t)$ is a dephasing function defined in Ref.~\cite{XLC18} (or see~\footnotemark[\value{footnote}]). Therefore, the space $\L_k$ is an invariant subspace of the system dynamics. 

The effective equation of motion corresponding to $\L_k$ reads
\begin{align}\label{eq-dep-n}
\dot{\rho}_k=\sum_{m,n=0}^{k-1} \dot{\eta}((m-n)t)\ketbra{m}{m} \rho_k \ketbra{n}{n}\, ,
\end{align}
where $\rho_k$ is a density matrix in $\L_k$.
Equation~\eqref{eq-dep-n} can be transformed to the standard form~\cite{HCLA14}
\begin{align}\label{eq:semi-cano}
\dot{\rho}_k=-\mathrm{i}[H_{\S}', \rho_k]+\sum_{p,q=1}^{k^2-1} d_{pq}\Big(G_p \rho_k G_q-\frac{1}{2}\l\{G_q G_p, \rho_k\r\}\Big)\, ,
\end{align}
in which $H_{\S}'=H_{\S}'^{\dagger}$; $\l\{G_p: p\in \{0,1,\cdots,k^2-1\}\r\}$ satisfies
\begin{align}
G_0=\frac{\mathbb{I}_k}{\sqrt{k}}\, , \: \: \: G_p=G_p^{\dagger}\, ,\: \: \: \operatorname{Tr}\l[G_p G_q\r]=\delta_{pq}\, ;
\end{align}
and
\begin{align}\label{eq-dec-mat-d}
d_{pq}=\sum_{m,n=0}^{k-1} \dot{\eta}((m-n)t) \bra{m} G_p \ket{m} \bra{n} G_q \ket{n}\, 
\end{align}
is the element of the {decoherence matrix} $\bm{d}_k$. (Hereafter, we shall abbreviate the time-dependence of  $\bm{d}_k(t)$, $\bm{d}_k^{\rm B}(t)$ and $\bm{D}_k(t)$.)

Choose a representation of $G_p$ and sort them as~\cite{BK08}
\begin{align}
\begin{split}
G_l^{\rm d}&=\frac{\operatorname{diag}\{1,\cdots,1,-l,0,\cdots,0\}}{\sqrt{l(l+1)}}   \qquad (1\leq l\leq k-1) \\
G_{mn}^{\rm s}&=\frac{1}{\sqrt{2}}(\ketbra{m}{n}+\ketbra{n}{m})   \qquad\quad (0\leq n < m\leq k-1) \\
G_{mn}^{\rm a}&=\frac{i}{\sqrt{2}}(\ketbra{m}{n}-\ketbra{n}{m})  \qquad\quad (0\leq n < m\leq k-1)
\end{split}
\end{align}
where the superscripts stand for ``diagonal'', ``symmetric'', and ``anti-symmetric'', respectively. Sort $G_l^{\rm d}$ by $l$, and  $G_{mn}^{\rm s}$ ($G_{mn}^{\rm a}$) by $m$ and $n$ in the ascending order, respectively. Then in $\bm{d}_k$, only the upper-left $(k-1) \times (k-1)$ block $\bm{d}^{\rm B}_k$ is non-trivial (all the other matrix elements are 0). Thus, $\bm{d}_k\geq 0$ is equivalent to $\bm{d}^{\rm B}_k\geq 0$. As a result, $\E_k(t,0)$ is divisible if and only if $\bm{d}^{\rm B}_k\geq 0$. Though complicated, $\bm{d}^{\rm B}_k$ can be transformed to a simpler form by the transformation $V_k\bm{d}^{\rm B}_k V_k^{\dagger}$, where $V_k$ is an invertible square matrix~\cite{H12}. By choosing $V_k$ in~\footnotemark[\value{footnote}], one obtains
\begin{align}\label{eq-cons-diag-mat}
\bm{D}_k=\left(\begin{array}{ccccc}
T_0& T_{-1} & T_{-2} &\cdots& T_{2-k}\\
T_1& T_0&T_{-1} &\ddots &\vdots\\
T_2& T_1&T_{0} &\ddots &T_{-2}\\
\vdots& \ddots&\ddots &\ddots &T_{-1}\\
T_{k-2}& \cdots&T_2 &T_1 &T_0\\
\end{array}\right),
\end{align}
in which $T_j=-\dot{\eta}((j+1)t)+2\dot{\eta}(jt)-\dot{\eta}((j-1)t)$. Thus, the divisibility condition becomes $\bm{D}_k\geq 0$, or equivalently, $\forall j\in\{2,3,\cdots,k \}$, $|\bm{D}_j|\geq 0$~\cite{A97}. 
	
Construct a chain of $\E(t,0)$-invariant subspaces ${\L}_{2}\subset {\L}_{3}\subset\cdots \subset {\L}_{\infty}={\L}_{\S}$, the corresponding divisibility conditions read $\bm{D}_2\geq0$, $\bm{D}_3\geq0$, $\cdots$, $\bm{D}_\infty\geq0$, respectively. Following Eq.~\eqref{eq-cons-diag-mat}, one has
\begin{align}
\bm{D}_m\geq0 \Longrightarrow \bm{D}_n\geq0 \quad(m>n)\,.
\end{align}
Therefore, the divisibility conditions form a hierarchy.

	
When $k=2$, $\bm{D}_k=[T_0]_{1\times1}$, with $T_0=-2 \operatorname{Re}\{\dot{\eta}(t)\}$. The divisibility condition is $\operatorname{Re}\{\dot{\eta}(t)\}\leq 0$. From the perspective of the density matrix, it reveals that the monotonic decay of the off-diagonal matrix elements guarantees the divisibility of two-level pure-dephasing systems. 
	
When $k>2$, $\bm{D}_k$ depends on both the real and the imaginary part of $\dot{\eta}(t)$. Consequently, the principal minors are also related to $\operatorname{Im}\l\{\dot{\eta}(t)\r\}$. That is, the monotonic decay of the off-diagonal matrix elements cannot guarantee the divisibility of the dynamics. 
	

By induction, $\E(t,0)$ is divisible if and only if ${\bm {D}}_{\infty}\geq 0$, or equivalently, $|\bm{D}_k|\geq 0$ ($k\geq2$). The spectrum of ${\bm{D}}_{\infty}$ is determined by the domain of the series $\sum_{k=-\infty}^{\infty}T_k \operatorname{e}^{i k \lambda}$ ($\lambda\in \mathbb{R}$)~\cite{G06}. Consequently, $\D_{\S}$ is divisible if and only if
\begin{align}
\sum_{k=-\infty}^{\infty}\dot{\eta}(k t) \operatorname{e}^{i k \lambda}\geq 0  \quad (\lambda\in \mathbb{R})\, .
\end{align}

Similar method can also be applied to other pure-dephasing dynamics. Usually, a set of non-Markovianity witnesses can be obtained and an exact divisibility condition can be derived. 
	
\emph{Example for the decay case.}---Consider an $N$-level system whose Hilbert space $\H_{\S}$ has an orthonormal basis $\{\ket{0},\ket{1},\cdots,\ket{N-1}\}$. The evolution of the system density matrix satisfies
\begin{align}\label{eq-decay}
\dot{\rho}_{\S}=\sum_{k=1}^{N-1}\gamma_k\l(-\frac{1}{2}\l\{\ketbra{k}{k}\!,\,\rho_{\S}\r\} +\bra{k}\!\rho_{\S}\!\ket{k}\rho^{(k)}\r)\,,
\end{align}
where $\gamma_k$ denotes 
decay rates which are allowed to be negative; $\rho^{(k)}$ is a density matrix in ${\L}_k$ (See Eq.~\eqref{eq-def-Lk} for the definition); $\gamma_k$ and $\rho^{(k)}$ are generally time-dependent. 
	
By spectrum decomposition, $\rho^{(k)}=\sum_{j=0}^{k-1}p_j^{(k)}\ketbra{k_j}{k_j}$, where $p_j^{(k)}\geq 0$, and $\ket{k_j}$ ($j\in\{0,1,\cdots,k-1\}$) form a time-dependent orthonormal basis of $\H_{k}=\operatorname{span}\{\ket{0}, \ket{1}, \cdots, \ket{k-1}\}$. Then, Eq.~\eqref{eq-decay} can be transformed to
\begin{align}\label{eq-decay-can}
\dot{\rho}_{\S}=\sum_{k=1}^{N-1}\sum_{j=0}^{k-1} p_j^{(k)} \gamma_k  \l( \sigma_{kj}^{-} \rho_{\S} \sigma_{kj}^{+}-\frac{1}{2} \l\{\sigma_{kj}^{+}\sigma_{kj}^{-}, \rho_{\S} \r\} \r)\,,
\end{align} 
where $\sigma_{kj}^{-} =\ketbra{k_j}{k}$ and $\sigma_{kj}^{+} =\ketbra{k}{k_j}$. 
Equation~\eqref{eq-decay-can} is in a {canonical form} and the divisibility condition reads $p_j^{(k)} \gamma_k\geq0$. Equivalently, the condition can also be expressed as $\gamma_k\geq0$ ($\forall k \in \{1,2,\cdots,N-1\}$). 

Now let us investigate the divisibility hierarchy. Let us prove that $\L_k$ are invariant spaces first. Denote the system's probability of being in $\ket{k}$ as $p_{k}$. Following Eq.~\eqref{eq-decay}, the time-evolution of the probabilities satisfy 
\begin{subequations}\label{eq-pro-evo}
	\begin{align}
	\dot{p}_{N-1}&=-\gamma_{N-1} p_{N-1}\, ,\label{}\\
	\dot{p}_{N-2}&=\rho^{(N-1)}_{N-2,N-2} \gamma_{N-1} p_{N-1}-\gamma_{N-2} p_{N-2}\, ,\label{}\\
    &\;\;\vdots \, \nonumber \\
	\dot{p}_{0}&=\rho^{(N-1)}_{0,0} \gamma_{N-1} p_{N-1}+\rho^{(N-2)}_{0,0} \gamma_{N-2} p_{N-2}+\cdots \nonumber
	\\ 
	&\quad+\rho^{(1)}_{0,0} \gamma_{1} p_{1}\, .\label{}
	\end{align}
\end{subequations}
Therefore, if $p_{N-1}(0)=0$, $p_{N-1}(t)\equiv0$; if $p_{N-1}(0)=p_{N-2}(0)=0$, $p_{N-1}(t)=p_{N-2}(t)\equiv0$; and so forth. 
That is, if $\rho_{\S}(0)\in {\L}_{N-1}$, $\rho_{\S}(t)\in {\L}_{N-1}$; if $\rho_{\S}(0)\in {\L}_{N-2}$, $\rho_{\S}(t)\in {\L}_{N-2}$; and so forth. So the maps $\E_k(t,0)$ corresponding to $\L_k$ are trace-preserving, thus all $\L_k $ are $\E(t,0)$-invariant for $t\geq0$~\footnotemark[\value{footnote}]. 
Equivalently, ${\L}_2\subset {\L}_3 \subset \cdots \subset {\L}_N$ form a chain of subspaces, and their divisibility conditions form a hierarchy.

In this example, the divisibility conditions of the subdynamics can be derived explicitly. The equation of motion corresponding to $\L_n$ reads
\begin{align}\label{eq-decay-m}
\dot{\rho}_n=\sum_{k=1}^{n-1}\sum_{i=0}^{k-1} p_i^{(k)} \gamma_k  \l( \sigma_{ki}^{-} \rho_n \sigma_{ki}^{+}-\frac{1}{2} \l\{\sigma_{ki}^{+}\sigma_{ki}^{-}, \rho_n \r\} \r)\,.
\end{align}
The equation is in the canonical form. Thus, the divisibility condition can be easily derived, which reads $\gamma_k\geq 0$ ($\forall k\in \{1,2,\cdots,n\}$). Obviously, the conditions for all the $n$ form a hierarchy.

\emph{Discussion on the non-trace-preserving subdynamics.}---For some decomposition $\H_{\S}=\H_{{\a_1}}\oplus\H_{\b}$, $\mathcal{E}_{{\a_1}{\a_1}}(t,0)$ is not trace-preserving. Thus, $\D_{{\a_1}}$ is not a legitimate quantum dynamics. In this case, if $\D_{\S}$ is divisible, we have the following proposition. 
\newtheorem{proposition}{Proposition}
\begin{proposition}\label{prop1}
Suppose the system dynamics is divisible but $\L_{{\a_1}}$ is not an invariant space. Then, in general, 
\begin{align}\label{eq-pro}
\mathcal{E}_{{\a_1}{\a_1}}(t_2,0)\neq \mathcal{Q}_{{\a_1}{\a_1}}(t_2,t_1)\mathcal{E}_{{\a_1}{\a_1}}(t_1,0)\, .
\end{align}
\end{proposition}
\begin{proof}
From Eqs.~\eqref{eq-def-div} and~\eqref{eq-def-sub1},
\begin{align}
\mathcal{E}_{{\a_1}{\a_1}}(t_2,0)=\mathcal{P}_{{\a_1}}\mathcal{Q}(t_2,t_1)(\mathcal{P}_{{\a_1}}+\mathcal{P}_{{\L}_{{\a_1}}^{\perp}})\mathcal{E}(t_1,0)\mathcal{P}_{{\a_1}}\, .
\end{align}
Therefore,
\begin{align}\label{Eq:prop1 map}
\mathcal{E}_{{\a_1}{\a_1}}(t_2,0)= \mathcal{Q}_{{\a_1}{\a_1}}(t_2,t_1)\mathcal{E}_{{\a_1}{\a_1}}(t_1,0)+ \mathcal{Q}_{{\a_1} {\a_1}^{\perp}}(t_2,t_1)\mathcal{E}_{{\a_1}^{\perp} {\a_1}}(t_1,0)\, .
\end{align}
Because $\mathcal{E}_{{\a_1}{\a_1}}(t_1,0)$ is not trace-preserving, $\mathcal{E}_{{\a_1}^{\perp}{\a_1}}(t_1,0)\neq 0$. Therefore, one has Eq.~\eqref{eq-pro}.
\end{proof}
	
	
The proposition reveals that when ${\L}_{{\a_1}}$ is not an invariant subspace of $\E(t,0)$, $\E_{\a_1 \a_1}(t,0)$ in general does not obey the composition relation in Eq.~\eqref{eq-comp-alpha}. In this case, the divisibility structure similar to that in the trace-preserving case does not exist. 

\emph{Conclusion.}---By decomposing open quantum systems with a new approach, i.e., the direct sum decomposition, we study the characteristics of the subsystem dynamics and their relationship with the dynamics of the original system. It is shown that 
if a chain of invariant subspaces exists, then the divisibility conditions form a hierarchy, offering us a systematic way of obtaining divisibility witnesses. With this approach, we study two classes of dynamics, i.e., the pure-dephasing processes and the decay processes. 

As divisibility is related to memory effects in open quantum systems, our results offer a systematic way of obtaining non-Markovianity witnesses. Moreover, as a new approach of decomposing the dynamics of open quantum systems, it is worthy of further study. 

We thank Wei-Min Zhang, Yi-Zheng Zhen, Zhe-Yan Wan and Wen-Fei Cao for helpful discussions. F.L.X. and Z.B.C. were supported by the National Natural Science Foundation of China (Grant No. 61125502) and the CAS. 

\appendix

\begin{widetext}

	\section{Lemmas required in the body of the text}
	\newtheorem{lemma}{Lemma}
	\begin{lemma}\label{lemma1}
		If $\E$ is CP, then $\mathcal{E}_{\a \a}=\P_{\a \a}\E \P_{\a \a}$ is also CP.
	\end{lemma}
	\begin{proof}
		Assume that there exists an ancillary $\A$, which corresponds to a space $\L_{\A}=\H_{\A}\times\H_{\A}$, with $\H_{\A}$ denoting its Hilbert space. The superoperator $\E$ ($\E_{\a \a}$) is CP, if and only if $\forall \A$ and $X_{\S\-\A}\geq0$ ($X_{\a\-\A}\geq0$), $\E \otimes \P_{\A \A}\l[X_{\S\-\A}\r]\geq 0$ ($\E_{\a \a} \otimes \P_{\A \A}\l[X_{\a\-\A}\r]\geq 0$), where $X_{\S\-\A}\in \L_{\S}\otimes \L_\A$ ($X_{\a\-\A}\in \L_{\a}\otimes \L_\A$), and $\P_{\A \A}$ is the projection superoperator of $\L_\A$~\cite{LHW17}. 
		
		Because $\mathcal{E}_{\a \a}=\P_{\a \a} \mathcal{E}_{\rm S} \P_{\a \a}$,
		\begin{align} 
		\E_{\a \a} \otimes \P_{\A \A}=\l(\P_{\a \a}\otimes\P_{\A \A}\r)\l(\E \otimes\P_{\A \A}\r)\l(\P_{\a \a}\otimes\P_{\A \A}\r)\, .
		\end{align}
		Thus, $\forall X_{\a\-\A}\in \L_{\a\-\A}$, 
		\begin{align} 
		\l(\E_{\a \a} \otimes \P_{\A \A}\r)\l[X_{\a\-\A}\r]=\l(\P_{\a \a}\otimes\P_{\A \A}\r)\l(\E \otimes\P_{\A \A}\r)\l[X_{\a\-\A}\r]\, .
		\end{align}
		Because $\E$ is CP, $\E \otimes\P_{\A \A}$ is a positive map. That is, $\forall\, X_{\a\-\A}\geq 0 $, $\l(\E \otimes\P_{\A \A}\r)[X_{\a\-\A}]\geq 0$. Moreover, $\P_{\a \a}\otimes\P_{\A \A}$ is also a positive map, so
		\begin{align}
		\l(\E_{\a \a} \otimes \P_{\A \A}\r)\l[X_{\a\-\A}\r]=\l(\P_{\a \a}\otimes\P_{\A \A}\r)\l(\E \otimes\P_{\A \A}\r)\l[X_{\a\-\A}\r]\geq0\, .
		\end{align}
		That is, $\E_{\a \a}$ is CP.
	\end{proof}
	
	\begin{lemma}\label{lemma2}
		All $\E_{\a \a}(t,0)=\P_{\a \a}\E(t,0) \P_{\a \a}$ are CPTP maps if and only if $\L_{\a}$ is $\E(t,0)$-invariant for $t\geq0$.
	\end{lemma}
	\begin{proof}
		If $\L_{\a}$ is an invariant subspace of $\E(t,0)$, for $\rho_{\S}(0)=\rho_{\a\a}(0)$,
		\begin{align}
		\rho_{\S}(t)=\E(t,0)\rho_{\S}(0)\in\L_{\a}\, .
		\end{align} 
		Thus,
		\begin{align}
		\rho_{\a\a}(t)=\E_{\a\a}(t,0)\rho_{\a\a}(0)=\rho_{\S}(t)\,.
		\end{align}
		That is, $\E_{\a\a}(t,0)$ is trace-preserving. As a result, $\E_{\a\a}(t,0)$ are CPTP maps. 
		
		Suppose all $\E_{\a \a}(t,0)$ are CPTP maps, but simultaneously $\L_{\a}$ is not $\E(t,0)$-invariant. By definition, there exist $X_{\a\a}$ and $t$ satisfying that $\E(t,0)X_{\a\a} \notin \L_{\a}$.
		However, if for all legitimate density matrix $\rho_{\a\a}(0)$,
		$\E(t,0)\rho_{\a\a}(0) \in \L_{\a}$,
		then $\forall \, X_{\a\a}\in\L_{\a}$,
		$\E(t,0)X_{\a\a} \in \L_{\a}$. 
		Therefore, there must exist some $\rho_{\a\a}(0)$ satisfying $\E(t,0)\rho_{\a\a}(0)\notin\L_{\a}$. Consequently, $\E_{\a\a}(t,0)$ cannot be trace-preserving, which contradicts the assumption that $\E_{\a \a}(t,0)$ are CPTP maps. 
	\end{proof}

Now we prove that under the conditions $\mathcal{E}_{\a_1^{\perp}\a_1}=0$ and $\operatorname{dim}(\L_{\a_1 })<\infty$, the existence of $\E^{-1}$ implies the existence of $\mathcal{E}_{\a_1\a_1}^{-1}$ and $\mathcal{E}_{\a_1^{\perp}\a_1^{\perp}}^{-1}$. For simplicity of notation, we define $\bm{A}=\mathcal{E}_{\a_1\a_1}$, $\bm{B}=\mathcal{E}_{\a_1^{\perp}\a_1^{\perp}}$ and $\bm{C}=\mathcal{E}_{\a_1 \a_1^{\perp}}$. Then $\E = \left(\begin{array}{c c}
	\bm{A} & \bm{C} \\
	\bm{0} & \bm{B} \\
	\end{array}\right)$. Accordingly, we express $\E ^{-1}$ also in the form of block matrix that $\E^{-1}=\left(\begin{array}{c c}
	\bm{D} & \bm{F} \\
	\bm{G} & \bm{E} \\
	\end{array}\right)$, where $\bm{D}=\left(\E^{-1} \right)_{\a_1 \a_1} $, $\bm{E}=\left(\E^{-1} \right)_{\a_1^{\perp} \a_1^\perp} $, etc. 
	By definition, we have
	\begin{align}
	\left(\begin{array}{c c}
	\bm{D} & \bm{F} \\
	\bm{G} & \bm{E} \\
	\end{array}\right)
	\left(\begin{array}{c c}
	\bm{A} & \bm{C} \\
	\bm{0} & \bm{B} \\
	\end{array}\right) & = \left(\begin{array}{c c}
	\bm{DA} & \bm{DC+FB} \\
	\bm{GA} & \bm{GC+EB} \\
	\end{array}\right) =\left(\begin{array}{c c}
	\bm{I} & \bm{0} \\
	\bm{0} & \bm{I} \\
	\end{array}\right)\,. \label{eq:lem42}\\
	\left(\begin{array}{c c}
	\bm{A} & \bm{C} \\
	\bm{0} & \bm{B} \\
	\end{array}\right) \left(\begin{array}{c c}
	\bm{D} & \bm{F} \\
	\bm{G} & \bm{E} \\
	\end{array}\right) &= \left(\begin{array}{c c}
	\bm{AD+CG} & \bm{AF+CE} \\
	\bm{BG} & \bm{BE} \\
	\end{array}\right) =\left(\begin{array}{c c}
	\bm{I} & \bm{0} \\
	\bm{0} & \bm{I} \\
	\end{array}\right)\,,\label{eq:lem41}  
	\end{align}
	According to Eq.~\eqref{eq:lem42}, $\bm{DA}=\bm{I}$. Because $\operatorname{dim}(\L_{\a_1 })<\infty$, $\bm{A}^{-1}=\bm{D}$. Through Eq.~\eqref{eq:lem42}, one also has $\bm{GA}=\bm{0}$. Multiplying both sides by $\bm{A}^{-1}$ on the right, one obtains $\bm{G}=\bm{0}$. Then through both Eqs.~\eqref{eq:lem42} and~\eqref{eq:lem41}, one can prove $\bm{BE}=\bm{EB}=\bm{I}$. Consequently, $\bm{B}^{-1}=\bm{E}$. Therefore, $\E_{\a_1 \a_1}^{-1}=\left(\E^{-1} \right)_{\a_1 \a_1}$ and $\E_{\a_1^\perp \a_1^\perp}^{-1}=\left(\E^{-1} \right)_{\a_1^\perp \a_1^\perp}$.
	Notice that in the above proof, $\operatorname{dim}\left(\L_{\rm S}\right)$ is not necessarily finite. Therefore, the theorem in the body text is applicable for infinite level systems under the condition that $\operatorname{dim}(\L_{\a_1})<\infty$. 
	
	\section{canonical form of master equations}
	Time-local master equations for finite-level quantum systems can all be expressed in the form
	\begin{align}\label{eq-time-loc}
	\dot{\rho} =\sum_k A_k(t) \rho  B_k^{\dagger}(t)\,,
	\end{align}
	where $A_k(t)$ and $B_k(t)$ denote operators in the system Hilbert space. For simplicity, we shall suppress their time dependence
	. Equation~\eqref{eq-time-loc} can be transformed to the form
	\begin{align}\label{eq:subcan}
	\dot{\rho} =-\mathrm{i}[H', \rho]+\sum_{p,q=1}^{n^2-1} d_{pq} \l(G_p \rho  G_q-\frac{1}{2}\l\{G_q G_p, \rho \r\}\r)\,,
	\end{align}
	where $H'=H'^{\dagger}$; $n$ is the dimension of the system; $\{G_p: p\in\{0,1,\cdots, n^2-1\}\}$ satisfies
	\begin{align}\label{eq:G_n}
	G_0=\frac{\bm{I}_n}{\sqrt{n}}\,, \: \: \: G_p=G_p^{\dagger}\,,\: \: \: \operatorname{Tr}\l[G_p G_q\r]=\delta_{pq}\,;
	\end{align}
	the elements
	\begin{align}\label{eq:deco}
	d_{pq} =\sum_k \operatorname{Tr}\l[G_p A_k\r] \operatorname{Tr}\l[G_q B_k^{\dagger}\r]\,.
	\end{align}
	form the \emph{decoherence matrix} $\bm{d} $ associated to the dynamics~\cite{HCLA14}.
	The decoherence matrix $\bm{d}$ is Hermitian, thus diagonalizable. Assume that $\operatorname{diag}\l\{\gamma_1,\gamma_2,\cdots,\gamma_{n^2-1}\r\}=U^{\dagger}(t) \bm{d}U(t) $, where $U(t)$ is a unitary operator. Then Eq.~\eqref{eq:subcan} can further be transformed to the \emph{canonical form}
	\begin{align}\label{eq:can}
	\dot{\rho} =-\mathrm{i}[H', \rho ]+\sum_{p=1}^{n^2-1} \gamma_p\l(L_p(t) \rho  L_p^{\dagger}(t)-\frac{1}{2}\l\{L_p^{\dagger}(t) L_p(t), \rho \r\}\r)\,,
	\end{align}
	where $L_p(t)=\sum_{q=1}^{n^2-1} U_{qp} G_q$. The $\{L_p(t): p\in\{1,2,\cdots,n^2-1\}\}$ form an orthonormal basis of traceless operators, i.e.,
	\begin{align}\label{eq:L_n}
	\operatorname{Tr}\l[L_p(t)\r]=0\,;\;\;\;\;\operatorname{Tr}\l[L_p(t)L_q(t)\r]=\delta_{pq}\,.
	\end{align}
	The dynamics is divisible if and only if $\gamma_p\geq0\; (\forall p \in \{1,2,\cdots,n^2-1\})$, 
	or equivalently, $\bm{d}\geq 0$~\cite{HCLA14}.

	\section{boson-boson pure-dephasing model and the decoherence matrices}
	
	In Ref.~\cite{XLC18}, we propose a boson-boson pure-dephasing model whose Hamiltonian reads
	\begin{align} 
	H=\omega_0 b^{\dagger}b+\sum _n \omega _n  b_n^{\dagger } b_n+\sum _n \lambda_n b^{\dagger}b b_n^{\dagger} b_n\, ,
	\end{align}
	where $\omega_0$, $b^{\dagger}$ and $b$ are the single particle energy, creation operator and annihilation operator of the bosonic system mode (S), respectively; $\omega _n$, $b_n^{\dagger }$ and $b_n$ are the single particle energy, creation operator and annihilation operator of the $k$th mode of the bosonic bath (B); $\lambda_n$ is the coupling strength between the system and the $k$th mode of the bath. 
	Suppose the initial state of the total system reads $\rho(0)=\rho_{\S}(0)\otimes e^{- H_{\rm B}/T } /\operatorname{Tr} \l\{ e^{- H_{\rm B}/T } \r\}$, where $\rho_{\S}(0)$ is the initial state of the system; 
	$H_{\rm B}=\sum _n \omega _n  b_n^{\dagger } b_n$ is the Hamiltonian of the bath; and $T$ is the temperature of the bath.  In the interaction picture, the equation of motion of the system can be formulated as
	\begin{align}\label{eq-dephasing}
	\dot{\rho}_{\S}=\sum_{m,n=0}^{\infty} \dot{\eta}((m-n)t)\ketbra{m}{m} \rho_{\S} \ketbra{n}{n}\, ,
	\end{align}
	where 
	\begin{align}
	\eta(t)=\sum_n {\ln\left(\frac{1-\operatorname{exp}\{- {\omega_n}/T\}}{1-\operatorname{exp}\{-{\omega_n}/T-i \lambda_n t\}}\right)}\, ,
	\end{align}
	and $\ket{m}$ and $\ket{n}$ stand for the number states of the system.

	If $\rho_{\rm S}(0)\in \L_{k}$, where ${\L}_k=\operatorname{span}\l\{\ketbra{m}{n}|0\leq m,n\leq k-1\r\}$, then the dynamics reduces to the pure-dephasing of a $k$-level system. Denote the density matrix as $\rho_k(t)$, then the equation of motion reads
	\begin{align}\label{eq-dep-n}
	\dot{\rho}_k=\sum_{m,n=0}^{k-1} \dot{\eta}((m-n)t)\ketbra{m}{m} \rho_k \ketbra{n}{n}\, ,
	\end{align}
	Following Eqs.~\eqref{eq-time-loc} and \eqref{eq:deco}, elements of the decoherence matrix ${\bm d}_k (t)$ satisfy
	\begin{align}
	d_{pq}(t)=\sum_{m,n=0}^{k-1} \dot{\eta}((m-n)t) \bra{m} G_p \ket{m} \bra{n} G_q \ket{n}\,.
	\end{align}
	Choose the following representation of $G_p$ that
	\begin{align}
	\begin{split}
	G_l^{\rm d}&=\frac{\operatorname{diag}\{1,\cdots,1,-l,0,\cdots,0\}}{\sqrt{l(l+1)}}   \qquad (1\leq l\leq k-1) \\
	G_{mn}^{\rm s}&=\frac{1}{\sqrt{2}}(\ketbra{m}{n}+\ketbra{n}{m})   \qquad\quad (0\leq n < m\leq k-1) \\
	G_{mn}^{\rm a}&=\frac{i}{\sqrt{2}}(\ketbra{m}{n}-\ketbra{n}{m})  \qquad\quad (0\leq n < m\leq k-1)
	\end{split}
	\end{align}
	where the superscripts stand for ``diagonal'', ``symmetric'', and ``anti-symmetric'', respectively.  
	Sort $G_l^{\rm d}$ by $l$, and $G_{mn}^{\rm s}$ ($G_{mn}^{\rm a}$) by $m$ and $n$ in the ascending order, respectively. Then in the decoherence matrix $\bm{d}_k$, only the upper-left $(k-1) \times (k-1)$ block $\bm{d}^{\rm B}_k$ is non-trivial.
	Explicitly, it reads
	\begin{align}\label{eq-d_k^B}
	\bm{d}^{\rm B}_k=\tilde{\bm G}_k {\dot{\bm \eta}}_k \tilde{\bm G}_k^{\dagger}\,,
	\end{align}
	where 
	\begin{align}
	{\dot{\bm \eta}}_k= \left(\begin{array}{ccccc}
	0& \dot{\eta}(-t) & \dot{\eta}(-2t) &\cdots& \dot{\eta}((1-k)t)\\
	\dot{\eta}(t)& 0&\dot{\eta}(-t) &\ddots &\vdots\\
	\dot{\eta}(2t)& \dot{\eta}(t)&0 &\ddots &\dot{\eta}(-2t)\\
	\vdots& \ddots&\ddots &\ddots &\dot{\eta}(-t)\\
	\dot{\eta}((k-1)t)& \cdots&\dot{\eta}(2t) &\dot{\eta}(t) &0\\
	\end{array}\right)\,,
	\end{align}
	and
	\begin{align}
	\tilde{\bm G}_k=\left(\begin{array}{cccccc}
	c_1& h_1 & 0 &\cdots& 0\\
	c_2& c_2&h_2 &\ddots &\vdots \\
	\vdots& \ddots&\ddots &\ddots &0 \\
	c_{k-1}& \cdots &\cdots & c_{k-1} & h_{k-1} \\
	\end{array}\right)\,,
	\end{align}
	with $c_{k-1}=\frac{1}{\sqrt{(k-1)k}}$ and $h_{k-1}=-\frac{k-1}{\sqrt{(k-1)k}}$. $\tilde{\bm G}_k$ can be further decomposed as 
	\begin{align}
	\tilde{\bm G}_k=\left(\begin{array}{cccc}
	c_1&       &        &      \\
	&  c_2  &        &      \\
	&       & \ddots &      \\
	&       &        &c_{k-1}  
	\end{array}\right) 
	\left(\begin{array}{cccc}
	1  &       &       &       \\
	1  &    1  &       &       \\
	\vdots&\ddots &\ddots &       \\
	1  & \cdots& 1     &   1  
	\end{array}\right)
	\left(\begin{array}{cccc}
	1   &       &        &      \\
	&  2    &        &      \\
	&       & \ddots &      \\
	&       &        &    {k-1}  
	\end{array}\right) 
	\left(\begin{array}{ccccc}
	1   &   -1  &        &      &     \\
	&    1  & \ddots &      &     \\
	&       & \ddots &   -1 &     \\
	&       &        &    1 &   -1
	\end{array}\right)\,.
	\end{align}
	Therefore, Eq.~\eqref{eq-d_k^B} can be transformed to
	\begin{align}
	\bm{d}^{\rm B}_k=V_k \bm{D}_k V_k^{\dagger}\,,
	\end{align}
	where
	\begin{align}
	V_k=\left(\begin{array}{cccc}
	c_1&       &        &      \\
	&  c_2  &        &      \\
	&       & \ddots &      \\
	&       &        &c_{k-1}  
	\end{array}\right) 
	\left(\begin{array}{cccc}
	1  &       &       &       \\
	1  &    1  &       &       \\
	\vdots&\ddots &\ddots &       \\
	1  & \cdots& 1     &   1  
	\end{array}\right)
	\left(\begin{array}{cccc}
	1   &       &        &      \\
	&  2    &        &      \\
	&       & \ddots &      \\
	&       &        &    {k-1}  
	\end{array}\right) \,,
	\end{align}
	and 
	\begin{align}\label{}
	\bm{D}_k=\left(\begin{array}{ccccc}
	T_0& T_{-1} & T_{-2} &\cdots& T_{2-k}\\
	T_1& T_0&T_{-1} &\ddots &\vdots\\
	T_2& T_1&T_{0} &\ddots &T_{-2}\\
	\vdots& \ddots&\ddots &\ddots &T_{-1}\\
	T_{k-2}& \cdots&T_2 &T_1 &T_0\\
	\end{array}\right),
	\end{align}
	with $T_j=-\dot{\eta}((j+1)t)+2\dot{\eta}(j t)-\dot{\eta}((j-1)t)$.

\end{widetext}

\bibliography{references}
	
\end{document}